\documentclass[11pt,a4paper]{article}
\usepackage{amsfonts}
\usepackage{comment}
\usepackage{mathtools}
\usepackage{subcaption}
\usepackage{amssymb}
\usepackage{url}
\usepackage{amsthm}
\usepackage[pdftex,dvipsnames]{xcolor}  
\usepackage{lineno}
\usepackage{bigstrut}
\usepackage{array}
\usepackage{stackengine}

\allowdisplaybreaks

\newtheorem{observation}{Observation}
\newtheorem{theorem}{Theorem}
\newtheorem{lemma}{Lemma}

\newtheorem{definition}{Definition}
\newtheorem{corollary}{Corollary}

\newcommand\xrowht[2][0]{\addstackgap[.5\dimexpr#2\relax]{\vphantom{#1}}}
\DeclareMathOperator{\arccot}{arccot}

\usepackage{authblk}

\begin{document}

\title{On approximating shortest paths in weighted triangular tessellations}

\date{}

\author[1]{Prosenjit Bose\thanks{Partially supported by NSERC. Email: jit@scs.carleton.ca}}
\author[1,2]{Guillermo Esteban\thanks{Partially supported by project PID2019-104129GB-I00/MCIN/AEI/ 10.13039/501100011033 and H2020-MSCA-RISE project 734922 - CONNECT. Email: g.esteban@uah.es}}
\author[2]{David Orden\thanks{Partially supported by project PID2019-104129GB-I00/MCIN/AEI/ 10.13039/501100011033 and H2020-MSCA-RISE project 734922 - CONNECT. Email: david.orden@uah.es}}
\author[3]{Rodrigo I. Silveira\thanks{Partially supported by project PID2019-104129GB-I00/MCIN/AEI/ 10.13039/501100011033 and H2020-MSCA-RISE project 734922 - CONNECT. Email: rodrigo.silveira@upc.edu}}

\affil[1]{School of Computer Science, Carleton University, Canada}
\affil[2]{Departamento de F\'{i}sica y Matem\'{a}ticas, Universidad de Alcal\'{a}, Spain}
\affil[3]{Departament de Matemàtiques, Universitat Politècnica de Catalunya, Spain}

\maketitle

\begin{abstract}
    We study the quality of weighted shortest paths when a continuous 2-dimensional space is discretized by a weighted triangular tessellation. In order to evaluate how well the tessellation approximates the 2-dimensional space, we study three types of shortest paths: a weighted shortest path~$ \mathit{SP_w}(s,t) $, which is a shortest path from $ s $ to $ t $ in the space; a weighted shortest vertex path $ \mathit{SVP_w}(s,t) $, which is an any-angle shortest path; and a weighted shortest grid path~$ \mathit{SGP_w}(s,t) $, which is a shortest path whose edges are edges of the tessellation.
    Given any arbitrary weight assignment to the faces of a triangular tessellation, thus extending recent results by Bailey et al.~[\textit{Path-length analysis for grid-based path planning}. Artificial Intelligence, 301:103560, 2021], we prove upper and lower bounds on the ratios $ \frac{\lVert \mathit{SGP_w}(s,t)\rVert}{\lVert \mathit{SP_w}(s,t)\rVert} $, $ \frac{\lVert \mathit{SVP_w}(s,t)\rVert}{\lVert \mathit{SP_w}(s,t)\rVert} $, $ \frac{\lVert \mathit{SGP_w}(s,t)\rVert}{\lVert \mathit{SVP_w}(s,t)\rVert} $, which provide estimates on the quality of the approximation.
     It turns out, surprisingly, that our worst-case bounds are independent of any weight assignment. Our main result is that $ \frac{\lVert \mathit{SGP_w}(s,t)\rVert}{\lVert \mathit{SP_w}(s,t)\rVert} = \frac{2}{\sqrt{3}} \approx 1.15 $ in the worst case, and this is tight. As a corollary, for the weighted any-angle path $ \mathit{SVP_w}(s,t) $ we obtain the approximation result $ \frac{\lVert \mathit{SVP_w}(s,t)\rVert}{\lVert \mathit{SP_w}(s,t)\rVert} \lessapprox 1.15 $.
\end{abstract}

\section{Introduction}

Geometric shortest path problems, where the goal is to find an optimal path between two points $ s $ and $ t $ in a geometric setting, are fundamental for variety of real-world applications. For example, autonomous navigation over different types of terrain is a building block of an intelligent vehicle~\cite{li2021pairwise,shen2022fast,wagner2015subdimensional}. Artificial intelligence is also an important part in the design of video games~\cite{kamphuis,sturtevant2}; developers usually design the movement of non-player characters following the edges of the cells that decompose the space. See Figure~\ref{fig:colossal}, which shows how 2D triangular cells are used in the strategy game ``Colossal Citadels''~\cite{ColossalCitadels}. Moreover, finding paths of minimum cost is one of the major features in geographic information science~\cite{floriani}.

\begin{figure}[tb]
		\centering
		\begin{subfigure}[tb]{0.45\textwidth}
            \includegraphics[scale=0.158]{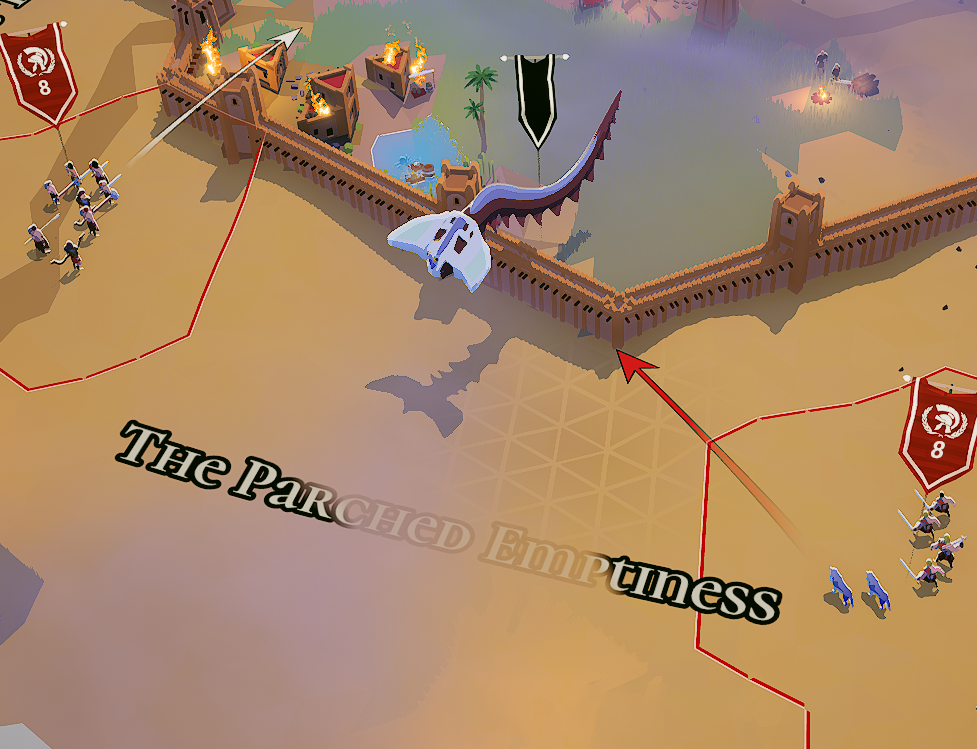}
     	\end{subfigure}
		\qquad
     	\begin{subfigure}[tb]{0.45\textwidth}
        	\includegraphics[scale=0.158]{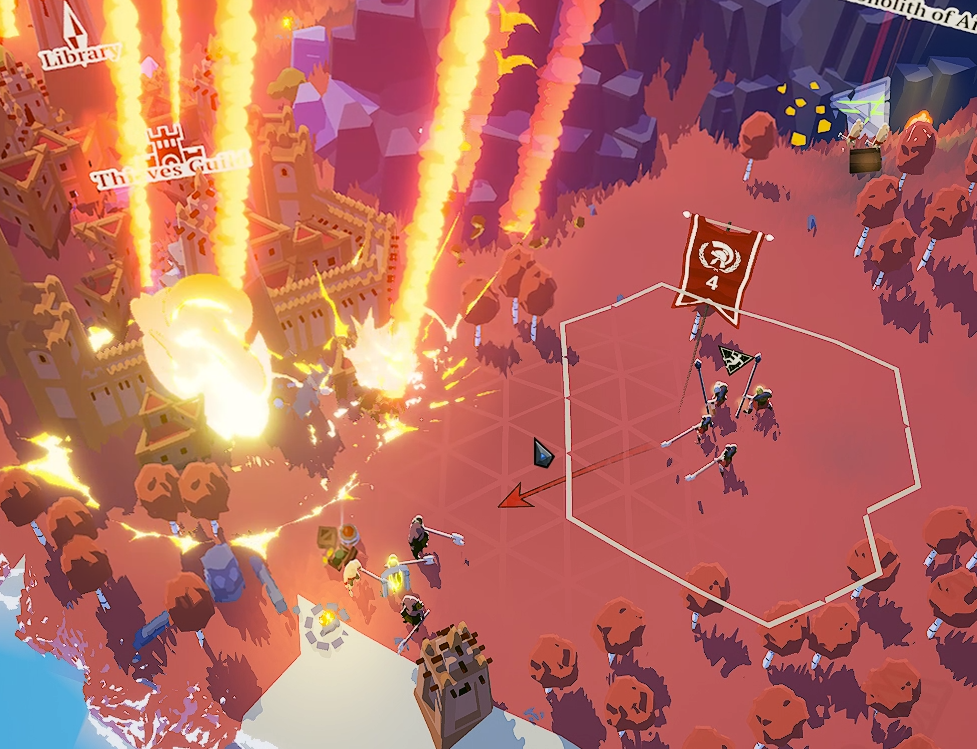}
     	\end{subfigure}
        \caption{Screenshots of the ``Colossal Citadels'' game by Uneven Dungeon. Used with permission from the author. Note the triangular grid underlying the scene.}
       	\label{fig:colossal}
	\end{figure}

One of the most general settings for geometric shortest path problems arises when the cost of traversing the plane varies depending on the region, that is, when the domain consists of a planar subdivision and each region $i$ of the subdivision has a weight~$\omega_i$, that represents the cost per unit of distance traveled in that region.

In gaming applications, this can be seen as an agent moving at different speeds when traversing a road, a dirt track, a forest, or a swamp area. Infrastructure planning takes into consideration planning, ecological and economic decision criteria. Hence, finding proper weights on a raster-based accumulated cost surface is crucial when placing power lines \cite{bachmann2018multi, hanssen2014least}, and pylons \cite{piveteau2017novel, santos2019optimizing}. It is also fundamental in the construction of highways and corridors \cite{seegmiller2021method}.
Thus, the cost of traversing a region is typically given by the Euclidean distance traversed in the region, multiplied by the corresponding weight.
The resulting metric is often called the \emph{weighted region metric}, and the problem of computing a shortest path between two points under this metric is known as the \emph{weighted region problem} (WRP)~\cite{mitchell1988algorithmic,Mitchell2}.

\subsection{Assumptions}

Applications usually require efficient and practical solutions for the WRP. Since an exact solution to the WRP is notoriously difficult, the problem is usually simplified in two ways.
First, the domain is approximated by using a (weighted) plane subdivision with a simpler structure.
Secondly, optimal shortest paths in that simpler subdivision are approximated.
The typical way to represent a 2D (or 3D) environment where shortest paths need to be computed is by using \emph{navigational meshes}~\cite{van2016comparative}.
These are polygonal subdivisions together with a graph that models the adjacency between the regions.
Path planning is then done first on the graph, to obtain a sequence of regions to be traversed, and then within each region, for which a shortest geometric path is extracted.

Triangles, convex polygons, disks or squares ---of different sizes--- are among the most frequently used region shapes~\cite{van2016comparative}. General navigational meshes allow efficient path planning in large environments as long as the region weights are limited to $\{1, \infty \}$ (i.e., free movement or obstacles), but when more weights $\{1, \omega_i,\ldots,\omega_j,\infty \}$ are needed (i.e., modeling different speeds for different types of ground), the complexity of computing the shortest path inside each region is most easily achieved through the use of the simplest possible navigational mesh: \emph{regular grids}.

In 2D, the only three types of regular polygons that can be used to tessellate continuous environments are triangles, squares and hexagons. The drawback with a grid is that it imposes a fixed resolution, requiring in general a large number of cells or regions. Still, grids are often used as navigational meshes (even for the simpler case of weights $\{1, \infty \}$), since they are easy to implement, are a natural choice for environments that are grid-based by design (e.g., many game designs, some robotic settings), and popular shortest path algorithms such as $A^*$ can be optimized for grids~\cite{harabor2016optimal,nagy,nash2007theta}.

\subsection{Definitions and notations}

Even when a regular grid is used as a navigational mesh, in practice, computing an exact weighted shortest path $\mathit{SP_w}(s,t)$ is difficult and, in fact, no exact algorithm exists for the WRP~\cite{Lou}: instead, in practice, one usually resorts to  approximations, by computing shortest paths on a weighted graph associated to the grid~\cite{Aleksandrov,Aleksandrov2,Aleksandrov3,ChengJV15}.

To this end, two different graphs have been considered in the literature~\cite{bailey2021path,kramm2018suboptimality,Nash,nash2007theta}, the \emph{corner-vertex graph} $ G_{\text{corner}} $ and the \emph{$ k $-corner grid graph} $ G_{k\text{corner}} $.

In $ G_{\text{corner}} $, the vertex set is the set of corners of the tessellation and every pair of vertices is connected by an edge. This graph is the complete graph over the set of vertices. Figure~\ref{fig:vertex-corner} depicts some of the neighbors of a vertex $ v $ in the corner-vertex graph. Note that in this graph some edges overlap.
A path in this graph is called a \emph{vertex path} or \emph{any-angle path}; a shortest vertex path between~$s$ and $t$ will be denoted by $\mathit{SVP_w}(s,t)$, where the subscript $w$ highlights that this path depends on a particular weight assignment $w$.

In $ G_{k\text{corner}} $, which is a subgraph of a corner-vertex graph, the vertex set is the set of corners of the tessellation, and each vertex is connected by an edge to a predefined set of $ k $ neighboring vertices, depending on the tessellation and other design decisions. See Figure~\ref{fig:3corner} for the $ 6 $-corner grid graph in a triangular tessellation. (Analogous $ k $-corner grid graphs can be defined for square and hexagonal tessellations.)
A path in this graph is called a \emph{grid path}; a shortest grid path between $s$ and $t$ will be denoted by $\mathit{SGP_w}(s,t)$.
	
	\begin{figure}[tb]
		\centering
		\begin{subfigure}[b]{0.44\textwidth}
            \includegraphics{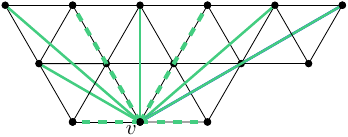}
     	    \caption{Some neighbors of a vertex $ v $ in $ G_{\text{corner}} $.}
     	    \label{fig:vertex-corner}
     	\end{subfigure}
		\qquad\qquad
     	\begin{subfigure}[b]{0.42\textwidth}
     	    \centering
        	\includegraphics{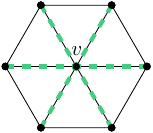}
         	\caption{All neighbors of a vertex $ v $ in $ G_{6\text{corner}} $.}
         	\label{fig:3corner}
     	\end{subfigure}
        \caption{Vertex $ v $ is connected to its neighbors in a triangular tessellation. The dashed lines represent the edges of the graphs that coincide with the edges of the cells.}
       	\label{fig:6neigh}
	\end{figure}
	
In all cases, the weight of each graph edge is defined by a function of the weights of the regions that the line segment associated with the edge traverses.
More formally, let $ T_i $ be a region in a subdivision with weight $ \omega_{i} \in \mathbb{R}_{\geq 0} $.
The cost of a segment $ \pi $ in the interior of a cell $ T_i $ is given by~$ \omega_{i}\rVert\pi\lVert $, where~$ \rVert \cdot \lVert $ is the Euclidean norm. In the case where~$ \pi $ lies on the boundary of two cells~$ T_j $ and~$ T_k $, the cost is~$ \min\{\omega_{j}, \omega_k\}\rVert\pi\lVert $. Thus, the weighted length of a path $ \Pi $ is the sum of the weighted lengths of its subpaths through each face and along each edge. With a slight abuse of notation, we still denote this by~$ \lVert \Pi \rVert$.

Figure~\ref{fig:comparison-corner} shows an example, illustrating the three paths considered in this work: the shortest path $ \mathit{SP_w}(s,t) $ (blue), the shortest vertex path $ \mathit{SVP_w}(s,t) $ (green), and the shortest grid path $ \mathit{SGP_w}(s,t) $ (red) in a $ 6 $-corner grid graph.
Note that, in the remainder of this work, any cell that is not depicted in the figures is considered to have infinite weight. In addition, if two paths coincide in a segment, one of them is depicted with dashed lines in that segment.

    \begin{figure}[tb]
    	\centering
        \includegraphics{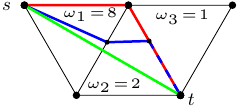}
    	\caption{$ \mathit{SP_w}(s,t) $ (blue), $ \mathit{SVP_w}(s,t) $ (green), and a $ \mathit{SGP_w}(s,t) $ (red) between two corners~$ s $ and $ t $ in $ G_{6\text{corner}} $. The cost of each path is~$ 16.75 $, $ 17.32 $ and $ 18 $, respectively, for a cell side length of $ 2 $.}
    	\label{fig:comparison-corner}
	\end{figure}
	
\subsection{Quality bounds for approximation paths}

The goal of this work is to understand the relation between $ \mathit{SGP_w}(s,t)$, $\mathit{SVP_w}(s,t) $, and $ \mathit{SP_w}(s,t) $, when a general weight assignment $\{1, \omega_i,\ldots,\omega_j,\infty \}$ is given in the WRP. Since $\mathit{SVP_w}(s,t) $ and $ \mathit{SGP_w}(s,t)$ are approximations of~$\mathit{SP_w}(s,t)$, a fundamental question is: what is the worst-case approximation factor that they can give?

In particular, we are interested in upper-bounding the ratios $ \frac{\lVert \mathit{SGP_w}(s,t)\rVert}{\lVert \mathit{SP_w}(s,t)\rVert} $ and $ \frac{\lVert \mathit{SVP_w}(s,t)\rVert}{\lVert \mathit{SP_w}(s,t)\rVert} $, since they indicate the approximation factor of the shortest grid path and shortest vertex path, respectively.
The ratio $ \frac{\lVert \mathit{SGP_w}(s,t)\rVert}{\lVert \mathit{SVP_w}(s,t)\rVert} $ is also studied, to see how different the two approximations can be.

The major contribution of this paper is the analysis of the quality of the three types of shortest paths for a weighted triangular grid for $ G_{6\text{corner}} $, which is the most natural graph defined on a triangular grid.

\subsection{Significance}

The WRP is very general, since it can be used to model many well-known variants of geometric shortest path problems.
Indeed, having all equal weights makes the metric equivalent to the Euclidean metric (up to scaling), while using weights $\{1,\infty\}$ allows to model paths amidst obstacles.

In the latter case, $\mathit{SGP_w}(s,t) $ and $\mathit{SP_w}(s,t) $ have been previously studied (note that $\mathit{SVP_w}(s,t) $ and $\mathit{SP_w}(s,t)$ coincide when the weights of the cells are taken in the set $\{1, \infty\}$). Algorithms using Snell's law of refraction, heuristic methods, or Dijkstra's algorithm are often used to find the shortest grid path~\cite{Nash,tran2020computing} between two given points. In case of large-scale grid environments some relaxed versions of Dijkstra and $A^*$ with linear running time $ O(n) $ ($ n $ is the size
of the grid) have been designed~\cite{Ammar}. Furthermore, some heuristics have been proposed for computing shortest paths in the context of game-programming \cite{Nash,yap2011any}, and for mobile robots \cite{carsten2009global,garcia2013dynamic,papadakis2013terrain}. Other algorithms have been suggested for isoline-based world representations \cite{gaw}, or for robots with two degrees of freedom \cite{Sharir}. In addition, some algorithms using heuristics, like Field D$^*$, have been generalized to 3D environments~\cite{carsten20063d} when computing $\mathit{SP_w}(s,t) $.

Almost all previous bounds on the ratio~$ \frac{\lVert \mathit{SGP_w}(s,t)\rVert}{\lVert \mathit{SP_w}(s,t)\rVert} $ consider a limited set of weights for the cells.
Bailey et al.~\cite{bailey2021path} considered only weights in the set $ \{1, \infty\} $ and proved that the weighted length of $ \mathit{SGP_w}(s,t) $ in hexagonal $ G_{6\text{corner}} $ and $ G_{12\text{corner}} $, square $ G_{4\text{corner}} $ and $ G_{8\text{corner}} $, and triangle $ G_{6\text{corner}} $ can be up to $ \approx\!1.15 $, $ \approx\!1.04, \ \approx\!1.41, \ \approx\!1.08 $, and $ \approx\!1.15 $ times the weighted length of $\mathit{SP_w}(s,t)$, respectively. In addition, for extended square grid neighborhoods such as $ G_{2^k\text{corner}} $ $r$-constrained it is proved that the length of an $ r$-constrained path is at most $ \frac{1}{\cos\left(\frac{\arccot(r)}{2}\right)} $ times the length of a shortest path~\cite{hew2017length}. For $ G_{2^k\text{corner}} $ and $ G_{2^k\text{center}} $ \cite{kramm2018suboptimality}, theoretical bounds for the ratio $ \frac{\lVert \mathit{SGP_w}(s,t)\rVert}{\lVert \mathit{SP_w}(s,t)\rVert} $ were presented, but no improvement was obtained over the results in~\cite{bailey2015path}.

Perhaps not surprisingly, the WRP turns out to be a challenging problem, so
the main challenge here is to obtain tight upper bounds that hold for \emph{any} assignment $\{1, \omega_i,\ldots,\omega_j,\infty \}$ of region weights. Efficient algorithms for the WRP only exist for a few special cases, e.g., rectilinear subdivisions with the~$L_1$ metric~\cite{ChenKT00}, or the \emph{maximum concealment problem}, where just regions with weights~$ 0 $ (travel in concealed free space), $1$ (travel in exposed free space), or $ \infty$ (travel through obstacles) are allowed~\cite{GewaliMMN90,mitchell1988algorithmic,Mitchell2}. This latter version of the WRP is related to stealth video games, such as ``Metal Gear'' \cite{MetalGear} or ``Assassin's Creed'' \cite{AssassinsCreed}, where the player uses stealth to avoid or overcome opponents, i.e., the objective is to minimize the time the moving agent is exposed to a given set of ``enemy'' observers.

The first algorithm for the WRP was a $(1+\varepsilon)$-approximation proposed by Mitchell and Papadimitriou~\cite{Mitchell2}, which runs in time $ O(n^8\log{\left(\frac{nNW}{w\varepsilon}\right)}) $, where~$ N $ is the maximum integer coordinate of any vertex of the subdivision, $ W $ and $ w $ are the maximum finite and the minimum nonzero integer weights assigned to the regions, respectively.
Substantial research has been devoted to designing faster approximation algorithms and studying different variants of the problem~\cite{Aleksandrov,Aleksandrov2,Aleksandrov3,rowe}.
Approximation schemes for the WRP are sophisticated methods that usually are based on variants of the continuous Dijkstra's algorithm, subdividing triangle edges in parts for which crossing shortest paths have the same combinatorial structure (e.g.,~\cite{Mitchell2}), or work by computing a discretization of the domain by carefully placing Steiner points (e.g., see~\cite{ChengJV15} for the currently best method of this type).
The lack of exact algorithms for the WRP is probably explained by the fact that it was recently shown to be impossible to solve this problem in the Algebraic Computation Model over the Rational Numbers~\cite{Lou}.
This is a model of computation where one can compute exactly any number that can be obtained from rational numbers by a finite number of basic operations.

\subsection{Results}

In this work, we consider tessellations where every face is an equilateral triangle (analogous ideas apply to square and hexagonal grids), and any arrangement of (non-negative) weights $\{1, \omega_i,\ldots,\omega_j,\infty \}$ to the cells of the discrete 2D environment. This extends recent results by Bailey et al.~\cite{bailey2021path}, who just considered weights in $\{1, \infty \}$.

Some advantages of triangular grids are that they can include hexagonal grids, and the distance between the vertices of adjacent cells is always the same, which simplifies distance calculations. In terms of computer games, movement of units in tight formation is allowed to have six directions and to turn smoothly. Furthermore, triangles can represent complex shapes, which is useful for building fortresses, bastions and streets, and interesting symmetrical shapes can be used for spells~\cite{ColossalCitadels}.

In contrast to previous work \cite{bailey2021path,hew2017length,kramm2018suboptimality}, we allow the weights $ \omega_i $ to take any value in $ \mathbb{R}_{\geq 0} $. When the weights of the cells are allowed to be arbitrary non-negative numbers, the only result that we are aware of is for square tessellations and another type of shortest path, with vertices at the center of the cells, for which Jaklin~\cite{Bound3} showed that $ \frac{\lVert \mathit{SGP_w}(s,t)\rVert}{\lVert \mathit{SP_w}(s,t)\rVert} \leq 2\sqrt{2} $. This latter model, considering vertices placed at the centers of the cells, simplifies collision avoidance during path execution, and has produced slightly different approximation results for the 2-dimensional terrain case, see~\cite{bailey2015path} for unweighted square tessellations. However, we do not study the ratios in this model since all of them are unbounded when we assign non-negative weights to the cells, see for instance Figure~\ref{fig:unbounded}.

\begin{figure}[tb]
		\centering
		\begin{subfigure}[t]{0.4\textwidth}
		\centering
            \includegraphics{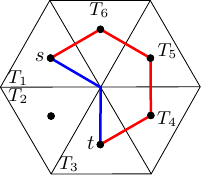}
     	    \caption{By setting the weight of $ T_2 $ to infinity, $ \mathit{SGP_w (s, t)} $ (red) must intersect some cells that $ \mathit{SP_w (s, t)} $ (blue) does not intersect.}
     	\end{subfigure}
		\qquad\qquad
     	\begin{subfigure}[t]{0.46\textwidth}
     	    \centering
        	\includegraphics{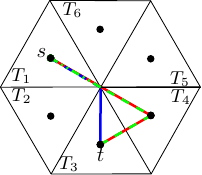}
         	\caption{Even if we increase the number of neighbors of each vertex, $ \mathit{SGP_w (s, t)} $ (red) and $ SVP_w(s,t) $ (green) intersect a cell that $ \mathit{SP_w (s, t)} $ (blue) does not intersect. Observe that the three paths coincide from $ s $ to the common vertex to $ T_1$ and $ T_4$.}
     	\end{subfigure}
        \caption{When the centers of the cells are used as the vertices of the associated graph, we can make the ratios $ \frac{\lVert \mathit{SGP_w}(s,t)\rVert}{\lVert \mathit{SP_w}(s,t)\rVert}, \frac{\lVert \mathit{SVP_w}(s,t)\rVert}{\lVert \mathit{SP_w}(s,t)\rVert} $ arbitrarily large by giving cells $ T_4, T_5, T_6 $ a finite weight much greater than 1, and cells $ T_1, T_3 $ weight $ 1 $.}
       	\label{fig:unbounded}
	\end{figure}

Our main result is that $ \frac{\lVert \mathit{SGP_w}(s,t)\rVert}{\lVert \mathit{SP_w}(s,t)\rVert} = \frac{2}{\sqrt{3}} $ in the worst case, for any (non-negative) weight assignment.
This implies bounds for the other two ratios considered.
Moreover, our upper bound for $ \frac{\lVert \mathit{SGP_w}(s,t) \rVert}{\lVert \mathit{SP_w}(s,t)\rVert}$ is tight, since it matches the lower bound of Nash~\cite{Nash}.
Table~\ref{tab:hexw} summarizes our results, including Nash's lower bounds.

\begin{table}[tb]
\begin{center}
\resizebox{\textwidth}{!}{\begin{tabular}{|c|c|c|}
\hline
\xrowht{5pt}             & Lower bound                      & Upper bound \\ \hline
\xrowht{17pt} $ \frac{\lVert \mathit{SGP_w}(s,t)\rVert}{\lVert \mathit{SP_w}(s,t)\rVert} $  & $ \frac{2}{\sqrt{3}} \approx 1.15 $ \cite{Nash} & $ \frac{2}{\sqrt{3}} \approx 1.15 $ (Thm.~\ref{thm:6})            \\ \hline
\xrowht{26pt} $ \frac{\lVert \mathit{SVP_w}(s,t)\rVert}{\lVert \mathit{SP_w}(s,t)\rVert} $  & $ \frac{2\sqrt{7\sqrt{3}-12}}{(7-4\sqrt{3})(6\sqrt{2}+\sqrt{7\sqrt{3}-12})} \approx 1.11 $  (Obs.~\ref{obs:5})   &  $ \frac{2}{\sqrt{3}} \approx 1.15 $  (Cor.~\ref{cor:5})      \\ \hline
\xrowht{17pt} $ \frac{\lVert \mathit{SGP_w}(s,t)\rVert}{\lVert \mathit{SVP_w}(s,t)\rVert} $ &   $ \frac{2}{\sqrt{3}} \approx 1.15 $    \cite{Nash} &   $ \frac{2}{\sqrt{3}} \approx 1.15 $  (Cor.~\ref{cor:6})   \\
\hline

\end{tabular}}
\end{center}
\caption{Bounds on the quality of approximations of shortest paths in weighted triangular tessellations for~$ G_{6\text{corner}} $.}
		\label{tab:hexw}
\end{table}

In order to obtain bounds on the ratios, we uncover some properties of the different paths that allow us to prove our approximation ratios. These properties are related to the behavior of shortest paths and to the geometry of a constant number of cells of the tessellation.
Surprisingly, a consequence of our analysis is that the worst-case ratios are upper-bounded by constants that are independent of the weights assigned to the regions in the tessellation, i.e., the assignment of arbitrary weights to the cells is not the determining factor on the worst-case ratio.

\section{$ \frac{\lVert \mathit{SGP_w}(s,t)\rVert}{\lVert \mathit{SP_w}(s,t)\rVert} $ ratio in $ G_{6\text{corner}} $ for triangular cells}

This section is devoted to obtaining, for two vertices $ s $ and~$ t $, an upper bound on the ratio~$ \frac{\lVert \mathit{SGP_w}(s,t)\rVert}{\lVert \mathit{SP_w}(s,t)\rVert} $ in~$ G_{6\text{corner}} $ in a triangular tessellation~$\mathcal{T}$ where faces are assigned arbitrary weights in $ \mathbb{R}_{\geq 0} $. We assume that $ \mathit{SP_w}(s,t) $ is unique, otherwise it is enough to repeat the following argument where we compute an upper bound for $ \frac{\lVert \mathit{SGP_w}(s,t)\rVert}{\lVert \mathit{SP_w}(s,t)\rVert} $. In addition, we suppose, without loss of generality, that the length of each edge of the triangular cells is~$ 2 $, in order to have a non-fractional length ($ \sqrt{3} $) for the cell height.

Let $ (s = u_1, u_2, \ldots, u_\ell = t) $ be the ordered sequence of consecutive points where a grid path $\mathit{GP_w}(s,t) $ and the shortest path $ \mathit{SP_w}(s,t) $ coincide; in the case where $\mathit{GP_w}(s, t) $ and~$ \mathit{SP_w}(s, t) $ share one or more segments, we define the corresponding points as the endpoints of each of these segments, see Figure~\ref{fig:7} for an illustration. Observation~\ref{thm:1} below is a special case of the mediant inequality.

	\begin{observation}
		\label{thm:1}	
		Let $\mathit{GP_w}(s, t) $ and $ \mathit{SP_w}(s,t) $ be, respectively, a weighted grid path, and a weighted shortest path, from $ s $ to $ t $. Let $ u_i $ and $ u_{i+1} $ be two consecutive points where $\mathit{GP_w}(s,t) $ and $ \mathit{SP_w}(s,t) $ coincide. Then, the ratio $ \frac{\lVert\mathit{GP_w}(s,t)\rVert}{\lVert \mathit{SP_w}(s,t)\rVert} $ is at most the maximum of all ratios $ \frac{\lVert\mathit{GP_w}(u_i,u_{i+1})\rVert}{\lVert \mathit{SP_w}(u_i,u_{i+1})\rVert}, i \in \{1, \ldots, \ell-1\} $.
	\end{observation}
	
\subsection{Crossing paths and weakly simple polygons}

The shape of the shortest paths~$ \mathit{SP_w}(s,t) $ and $ \mathit{SGP_w}(s,t) $ is unknown to us. Moreover, knowing the exact shape of $ \mathit{SP_w}(s,t) $ is difficult. In addition, the two paths might intersect many different cells of the tessellation, so we need to take into account the weights of all the intersected regions. To address all these issues, for a given $\mathit{SP_w}(s,t)$ we define a particular grid path called a \textit{crossing path}~$ X(s,t) $, whose structure is simpler than the structure of a $ \mathit{SGP_w}(s,t) $, and whose weighted length provides an upper bound on the weighted length of $ \mathit{SGP_w}(s,t) $. See the orange path in Figure~\ref{fig:7}. Thus, since $ \lVert \mathit{SGP_w}(s,t) \rVert \leq \lVert X(s,t) \rVert $, the key idea to prove an upper bound on the ratio $ \frac{\lVert \mathit{SGP_w}(s,t)\rVert}{\lVert \mathit{SP_w}(s,t)\rVert} $ is to upper-bound the ratio $ \frac{\lVert \mathit{X}(s,t)\rVert}{\lVert \mathit{SP_w}(s,t)\rVert} $. To do so, we analyze the components resulting from the intersection between $\mathit{SP_w}(s,t)$ and $ X(s,t)$. Each component is a weakly simple polygon, whose boundary consists of the portions of $\mathit{SP_w}(s,t)$ and $ X(s,t)$ between the intersection points. The reason the polygons are weakly simple is because some portion of the boundary may be shared between the two paths. These weakly simple polygons are the basic unit that we will analyze to obtain our main result. We also obtain a relation between the weights of some cells intersected by~$ \mathit{SP_w}(s,t) $ and $ X(s,t)$. Notice that, for one type of weakly simple polygon, we will need a finer analysis using \emph{shortcut paths} $ \Pi_i(s,t) $, which are defined in Section~\ref{sec:weakly} (see Definition~\ref{def:22}).

	\begin{figure}[tb]
		\centering
		\includegraphics[width=\textwidth]{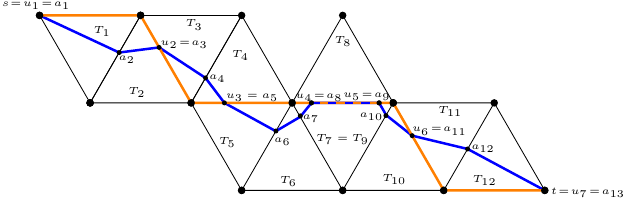}
		\caption{Weighted shortest path $ \mathit{SP_w}(s,t) $ (blue) and the crossing path $ X(s,t) $ (orange) from~$ s $ to $ t $ in a triangular tessellation.}
		\label{fig:7}
	\end{figure}
	
We begin by defining $ X(s,t) $. Let $ (T_1, \ldots, T_n) $ be the ordered sequence of consecutive triangular cells intersected by~$ \mathit{SP_w}(s, t) $ in the tessellation $ \mathcal{T} $. Let~$ v^i_1, v^i_2, v^i_3 $ be the three consecutive corners on the boundary of~$ T_i, \ 1 \leq i \leq n $, in clockwise order. Let $ (s = a_1, a_2, \ldots, a_{n+1} = t) $ be the sequence of consecutive points where~$ \mathit{SP_w}(s, t) $ changes the cell(s) it belongs to in $ \mathcal{T} $. In particular, let $ a_i $ and~$ a_{i+1} $ be, respectively, the points where $ \mathit{SP_w}(s,t) $ enters and leaves $ T_i $. Informally, we define $ X(s,t) $ based on $ \mathit{SP_w}(s,t) $. Essentially, $ \mathit{SP_w}(s,t) $ must traverse a number of cells of the tessellation. Although we do not know which cells are crossed, based on the geometry of the cell, we can compute an approximate path, $ X(s,t)$. We formalize this notion below:
	
	\begin{definition}
		\label{def:1}
		The \emph{crossing path} $ X(s,t) $ between two vertices $ s $ and $ t $ is defined by the sequence $ (X_1, \ldots, X_n) $, where $ X_i $ is a sequence of vertices determined by the pair~$ (a_i, a_{i+1}), \ 1 \leq i \leq n $, as follows:
		\begin{enumerate}
		    \setlength\itemsep{0em}
			\item If $ a_{i} $ and $ a_{i+1} $ are on the same edge $ e^i_1 \in T_i $, let $ v $ and $ w $ be the endpoints of $e^i_1$,
			where $a_i$ is encountered before $a_{i+1}$ when traversing $ e^i_1$ from $v$ to $w$.
			Then $ X_i =(v,w)$ if $ a_i = v$, and $ X_i=(w) $ otherwise, see Figure~\ref{fig:13}.
			\item If $ a_i $ is a corner of $ T_i $, and $ a_{i+1} $ belongs to the interior of the edge $ e^i_2 \in T_i $ not adjacent to $ a_i $, let $ p $ be the midpoint of $ e^i_2 $. If $ a_{i+1} $ is to the left of $ \overrightarrow{a_ip} $, $ X_i $ is $ a_i $ and the endpoint of $ e^i_2 $ to the left of $ \overrightarrow{a_ip} $, see Figure~\ref{fig:32}. Otherwise, $ X_i $ is $ a_i $ and the endpoint of $ e^i_2 $ to the right of $ \overrightarrow{a_ip} $.
			\item If $ a_{i} $ belongs to the interior of an edge $ e^i_1 \in T_i $ and $ a_{i+1} $ is a corner of $ T_i $ not contained in $ e^i_1 $, $ X_i = (a_{i+1}) $, see Figure~\ref{fig:44}.
			\item If $ a_i $ and $ a_{i+1} $ belong to the interior of two different edges $ e^i_1, e^i_2 \in T_i $, $ X_i $ is the common endpoint of $ e^i_1 $ and $ e^i_2 $, see Figure~\ref{fig:55}.
		\end{enumerate}
	\end{definition}
	
	\begin{figure}[tb]
	\captionsetup[sub]{justification=centering}
		        \centering
		            \begin{subfigure}[b]{0.2\textwidth}
			            \centering
	    	            \includegraphics{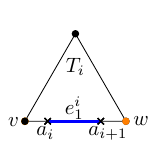}
	    	            \caption{}
	    	            \label{fig:13}
	                \end{subfigure}
    	           \quad
		            \begin{subfigure}[b]{0.2\textwidth}
			            \centering
	    	            \includegraphics{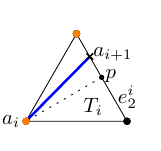}
	    	            \caption{}
	    	            \label{fig:32}
	                \end{subfigure}
    	            \quad
		            \begin{subfigure}[b]{0.2\textwidth}
   			            \centering
	    	            \includegraphics{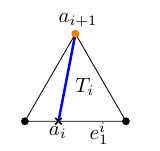}
	    	            \caption{}
	    	            \label{fig:44}
    	            \end{subfigure}
    	            \quad
		            \begin{subfigure}[b]{0.2\textwidth}
   			            \centering
	    	            \includegraphics{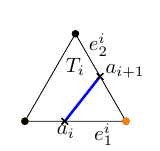}
	    	            \caption{}
	    	            \label{fig:55}
    	            \end{subfigure}
    	            \caption{Some of the positions of the intersection points between $ \mathit{SP_w}(s,t) $ (blue) and a cell. The vertices of the crossing path $ X(s,t) $ in a triangular cell are depicted in orange.}
            	\end{figure}
	
	
Note that the four conditions in Definition \ref{def:1} cover all possible cases. Thus, each shortest path $ \mathit{SP_w}(s, t) $ correspond to a single unique $ X(s, t) $.

Let $ (s\!=\!u_1, u_{2}, \ldots, u_\ell\!=\!t) $ be the sequence of consecutive points where~$ X(s, t) $ and~$ \mathit{SP_w}(s, t) $ coincide. The union of $ \mathit{SP_w}(s,t) $ and $ X(s,t) $ between two consecutive points $ u_j $ and $ u_{j+1} $, for $ 1 \leq j < \ell $, induces a weakly simple polygon (see~\cite{Chang} for a formal definition).
Henceforth, we use the term polygon to mean weakly simple polygon. We distinguish six different types of weakly simple polygons, denoted $P_1, \dots, P_6$, depending on the number of edges of $ \mathcal{T} $  intersected by~$ \mathit{SP_w}(u_j,u_{j+1}) $, see Figure~\ref{fig:37}.


    \begin{definition}
		\label{def:12}
		Let $ u_j $ and $ u_{j+1} $ be two consecutive points where $ X(s, t) $ and $ \mathit{SP_w}(s, t) $ coincide in a triangular tessellation. Let $ p $ be a corner of $ \mathcal{T} $ contained in~$ X(u_j, u_{j+1}) $. A weakly simple polygon induced by $ u_j $ and $ u_{j+1} $ is of \emph{type~$ P_k $}, for $ 1 \leq k \leq 6 $, if the subpath $ \mathit{SP_w}(u_j, u_{j+1}) $ intersects $ k $ consecutive edges around~$ p $.
	\end{definition}
	
	\begin{figure}[tb]
		\centering
		\includegraphics{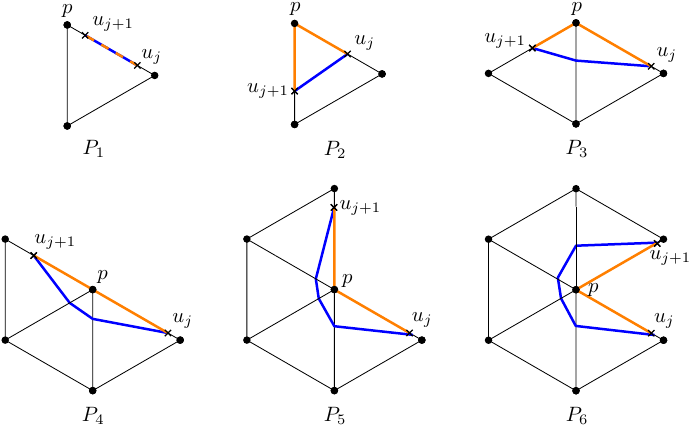}
		\caption{An example of weakly simple polygon for each type $ P_k $ (notice that, w.l.o.g., $ u_j $ and $ u_{j+1} $ could have other positions), and the subpath of the crossing path $ X(s,t) $ (orange) from~$ u_j $ to $ u_{j+1} $ intersecting consecutive triangular cells. $ \mathit{SP_w}(u_j, u_{j+1}) $ is depicted in blue.}
		\label{fig:37}
	\end{figure}
	
	The weakly simple polygons are an important tool in our proof, since it is enough to upper-bound $ \frac{\lVert X(s,t)\rVert}{\lVert \mathit{SP_w}(s,t)\rVert} $ for polygons of each type~$P_1, \dots, P_6$.
	
	\begin{lemma}
        \label{obs:uniqueness}
        The polygons from Definition~\ref{def:12} are the only weakly simple polygons that can arise.
    \end{lemma}

    \begin{proof}
        Suppose the point $ u_j$ is a vertex, then if the next point where~$ \mathit{SP_w}(s, t) $ changes cell belongs to the same edge as $ u_j $ then we have a polygon of type~$ P_1$, by Definition~\ref{def:1}, part~1. Otherwise, we obtain a polygon of type~$ P_k$, by Definition~\ref{def:1}, parts~2 and 4.

        Now, suppose $ u_j $ is a point on the interior of an edge. Let $ b $ be the next point where $ \mathit{SP_w}(s, t) $ changes the cell(s) it belongs to. Then if $ b $ belongs to the same edge as $ u_j $, we obtain a polygon of type~$ P_1$, by Definition~\ref{def:1}, part~1. Otherwise, there is a set of points where $ \mathit{SP_w}(s, t) $ coincides with $ k $ consecutive edges around a point~$ p $, which is the common endpoint of the edges containing $ u_{j} $ and $ b $. Hence, by Definition~\ref{def:1}, parts~3 and~4, $ X(s,t) $ intersects $ u_j, p $, and $ u_{j+1} $, which is the last point of the set. Thereby, defining polygons of type~$ P_2, \ldots, P_6$ depending on the value of~$ k$.
    \end{proof}
	
	\subsection{Bounding the ratio for weakly simple polygons}\label{sec:weakly}
	We are now ready to upper-bound the ratio $ \frac{\lVert X(u_j, u_{j+1})\rVert}{\lVert \mathit{SP_w}(u_j, u_{j+1}) \rVert} $ for each of the six types of weakly simple polygons in~$ G_{6\text{corner}} $.
	
	First, we make a geometric observation that will be needed later.
    Let $ p $ and~$ q $ be two points that are in the interior of two different edges on the boundary of the same triangular cell. Then, the length of the segment between $ p $ and~$ q $ is given in Observation~\ref{obs:3}, which can be proved using the law of cosines.

    \begin{observation}
		\label{obs:3}
		Let $ T_i $ be a triangular cell, and let $ (u, v, w) $ be the three vertices of $ T_i $, in clockwise order. Let $ p \in [u, v] $ and $ q \in [v, w] $ be two points on the boundary of $ T_i $, see Figure~\ref{fig:distance}. Then, $ |pq| = \sqrt{|pv|^2+|vq|^2-|pv||vq|} $.
	\end{observation}
	
	\begin{figure}[tb]
	    \centering
	    \includegraphics{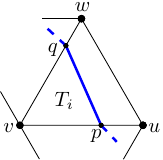}
	    \caption{Subpath of a weighted shortest path $ \mathit{SP_w}(s,t) $ between two points $ p $ and $ q $.}
	    \label{fig:distance}
    \end{figure}

    We observe that, by definition, we have $ \frac{\lVert X(u_j,u_{j+1})\rVert}{\lVert \mathit{SP_w}(u_j,u_{j+1})\rVert} = 1 $ for a polygon of type~$P_1$. Therefore, our focus will be on bounding polygons of type~$P_2, \dots, P_6$.
	We begin from the simpler case of polygons of type~$ P_3, \dots, P_6 $, and later we will consider a polygon of type~$P_2$, whose analysis is substantially more involved.
	
	\begin{lemma}
		\label{lem:27}
		Let $ u_j, \ u_{j+1} \in \mathcal{T} $ be two consecutive points where a shortest path~$ \mathit{SP_w}(s,t) $ and the crossing path $ X(s,t) $ coincide. If $ u_j, u_{j+1} $ induce a weakly simple polygon of type~$ P_k $, for $ 3 \leq k \leq 6 $, then $ \frac{\lVert X(u_j, u_{j+1})\rVert}{\lVert \mathit{SP_w}(u_j, u_{j+1}) \rVert} \leq \frac{2}{\sqrt{3}} $.
	\end{lemma}
	
	\begin{proof}
	    Let $ T_{i-1} \cap T_i $ be the cell boundary containing $ u_j $, and let $ T_{i+k-2} \cap T_{i+k-1} $ be the cell boundary containing $ u_{j+1} $, see Figure~\ref{fig:66} for an example with $ k = 3 $. Since $ u_j $ and $ u_{j+1} $ are two consecutive points where $ \mathit{SP_w}(s,t) $ and $ X(s,t) $ coincide, and they induce a polygon of type~$ P_k $, $ \mathit{SP_w}(s,t) $ enters $ T_i $ from $ T_{i-1} $ through $ u_j $, and $ \mathit{SP_w}(s,t) $ leaves $ T_{i+k-2} $, and enters $ T_{i+k-1} $ through $ u_{j+1} $. Let $ (v^i_1, p, v^i_2) $ be the sequence of consecutive vertices on the boundary of $ T_i $, in clockwise order, and $ (v^{i+k-2}_1, p, v^{i+k-2}_2) $ be the sequence of consecutive vertices on the boundary of $ T_{i+k-2} $, in clockwise order. Let $ x \in [v^i_2, p] $ be the point where $ \mathit{SP_w}(s,t) $ leaves $ T_i $, and let $ y \in [v^{i+k-2}_1, p] $ be the point where $ \mathit{SP_w}(s,t) $ enters $ T_{i+k-2} $. Let $ a, b, c, d $ be the lengths $ |u_jp|, |u_jx|, |u_{j+1}p| $, and $ |u_{j+1}y| $, respectively.
	
	    \begin{figure}[tb]
			\centering
	    	\includegraphics{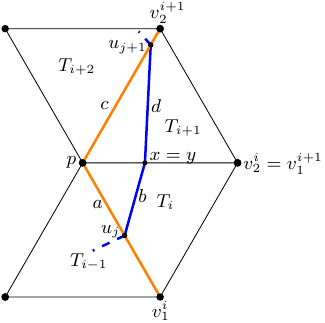}
    	    \caption{Subpaths of the crossing path $ X(s,t) $ (orange), and $ \mathit{SP_w}(s,t) $ (blue) traversing a polygon of type~$ P_3 $ in a triangular tessellation.}
    	    \label{fig:66}
        \end{figure}
	
	    An upper bound on the ratio $ \frac{\lVert X(u_j, u_{j+1})\rVert}{\lVert \mathit{SP_w}(u_j, u_{j+1})\rVert} $ in a polygon of type~$ P_k $ is given by
		\begin{align*}
			\frac{\lVert X(u_j, u_{j+1})\rVert}{\lVert \mathit{SP_w}(u_j, u_{j+1})\rVert} & = \frac{a\min\{\omega_{i-1}, \omega_i\} +c\min\{\omega_{i+k-2}, \omega_{i+k-1}\}}{b\omega_i+\lVert \mathit{SP_w}(x,y) \rVert + d\omega_{i+k-2}} \leq \frac{a\omega_i+ c\omega_{i+k-2}}{b\omega_i+d\omega_{i+k-2}} \leq \\
			& \leq \frac{a\omega_i+ c\omega_{i+k-2}}{\frac{\sqrt{3}}{2}a\omega_i+\frac{\sqrt{3}}{2}c\omega_{i+k-2}} = \frac{a\omega_i+ c\omega_{i+k-2}}{\frac{\sqrt{3}}{2}(a\omega_i+c\omega_{i+k-2})} = \frac{2}{\sqrt{3}}.
    	\end{align*}
	Note that the $ \lVert \mathit{SP_w}(x,y) \rVert $ term in the second denominator is zero for the case of $ P_3 $, while for $ P_k, k \geq 4 $, will just make the fraction smaller. Also, the last inequality comes from the fact that $ [p,u_j] $ is the side of a triangle adjacent to an angle of $ \frac{\pi}{3}$, and $ [u_j,x] $ is the side opposite to this angle. Hence, if we want to minimize the length of $ [u_j,x] $, it has to be perpendicular to $ [p,u_j] $, so $ \lvert u_jx \rvert \geq \lvert pu_j \rvert \sin{\frac{\pi}{3}} = \lvert pu_j \rvert \frac{\sqrt{3}}{2} $. An analogous reasoning can be applied in triangle $ T_{i+1} $ for distances $ c $ and $ d $.
	\end{proof}
	
	Next, we present a similar bound for a polygon of type~$ P_2 $. There is the added difficulty that for a polygon of type~$ P_2 $ it is possible to find an instance where $ \mathit{SP_w}(s,t) $ intersects a weakly simple polygon of type~$ P_2 $ such that the ratio $ \frac{\lVert X(s, t)\rVert}{\lVert \mathit{SP_w}(s, t) \rVert} $ is much larger than $ \frac{\lVert \mathit{SGP_w}(s, t)\rVert}{\lVert \mathit{SP_w}(s, t) \rVert} $, see Figure~\ref{fig:63}. However, between $ s $ and $ t $ there are other grid paths shorter than $ X(s,t) $ that intersect a polygon of type~$ P_2 $. In order to obtain an upper bound when $ \mathit{SP_w}(s,t) $ intersects a polygon of type~$P_2$, we need a finer analysis.
	
	\begin{figure}[tb]
			\centering
	    	\includegraphics{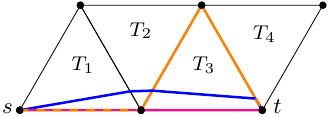}
	    	\caption{The weights of the cells are $ \omega_1\!=\!2, \ \omega_2\!=\!1.5, \ \omega_3\!=\!2 $, and $ \omega_4\!=\!1.2 $. The ratio $ \frac{\lVert X(s,t) \rVert}{\lVert \mathit{SP_w}(s,t)\rVert} $ is $ \approx\!1.19 $, whereas the ratio $ \frac{ \lVert \Pi_3(s,t) \rVert}{\lVert \mathit{SP_w}(s,t)\rVert} $ is almost $ 1 $.}
	    	\label{fig:63}
        \end{figure}
	
	In Definition~\ref{def:22} we define another class of grid paths, called \emph{shortcut paths}, that gives a tighter upper bound on the ratio $ \frac{\lVert \mathit{SGP_w}(s, t)\rVert}{\lVert \mathit{SP_w}(s, t) \rVert} $ when a weakly simple polygon of type~$ P_2$ is intersected by $ \mathit{SP_w}(s,t) $. See the purple path in Figure~\ref{fig:62}.
	
	\begin{figure}[tb]
	    \centering
	    \includegraphics{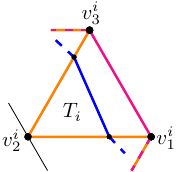}
	    \caption{Weighted shortest path $ \mathit{SP_w}(s,t) $ (blue), crossing path $ \mathit{X}(s,t) $ (orange), and shortcut path $ \Pi_i(s,t) $ (purple) intersecting a weakly simple polygon of type~$ P_2 $.}
	    \label{fig:62}
    \end{figure}
		
		
		\begin{definition}
			\label{def:22}
			Let $ \mathit{SP_w}(s,t) $ enter and leave cell $ T_i \in \mathcal{T} $ through the edges $ [v_1^i, v_2^i] $ and $ [v_2^i, v_3^i] $, respectively. If $ X(s,t) $ contains the subpath $ (v_1^i, v_2^i,v_3^i) $, the \emph{shortcut path} $ \Pi_i(s,t) $ is defined as the grid path $ X(s, v_1^i) \cup (v_1^i,v_3^i) \cup X(v_3^i, t) $.
		\end{definition}
		
	By using the shortcut path $ \Pi_i(s,t) $, we obtain a relation between the weights of the cells adjacent to $ T_i \in \mathcal{T} $ intersected by the crossing path $ X(s,t) $. Before obtaining an upper bound on the ratio $ \frac{\lVert X(s,t)\rVert}{\lVert \mathit{SP_w}(s,t)\rVert} $ for a polygon of type~$ P_2 $, we define a \emph{$ P_2 $-triple} of cells, see Figure~\ref{fig:29}, that will be useful later.
	
	\begin{definition}
	    \label{def:triple}
	    A \emph{$ P_2$-triple} between two vertices $ s $ and $ t $ is defined as a set of five consecutive cells $ T_1, \ldots, T_5 $ with the following properties:
	    \begin{itemize}
	    \setlength\itemsep{0em}
	        \item The cells form a strip of width $ \sqrt{3} $.
	        \item $ s $ is the vertex common to $ T_1 $ and $ T_2 $ not adjacent to $ T_3 $.
	        \item $ t $ is the vertex common to $ T_4 $ and $ T_5 $ not adjacent to $ T_3 $.
	        \item $ \mathit{SP_w}(s,t) $ determines three weakly simple polygons of type~$ P_2$, one around each of the vertices of~$T_3$.
	    \end{itemize}
	\end{definition}
	
	\begin{figure}[tb]
	    \centering
	    \includegraphics{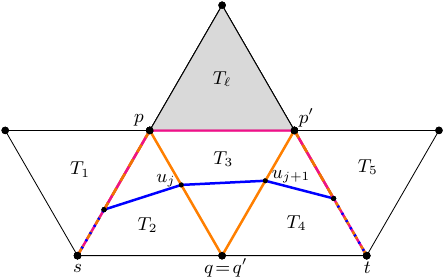}
	    \caption{$ P_2 $-triple between $ s $ and $ t $ is depicted in white. $ \mathit{SP_{w}}(s,t), \ \Pi_3(s,t) $ and $ X(s,t) $ are depicted in blue, purple and orange, respectively. Observe that the three paths coincide when $ \mathit{SP_{w}}(s,t)$ coincide with the edges of the cells.}
	    \label{fig:29}
	\end{figure}
	
	By using $ P_2$-triples we reduce the number of cases to analyze. Using Observation~\ref{thm:1}, Lemma~\ref{lem:27}, and the definition of polygons~$P_1$, we know that the only type of weakly simple polygon that could give a ratio larger than $ \frac{2}{\sqrt{3}} $ is a $ P_2 $.
	Thus, we can assume that the shortest paths giving the maximum ratio intersect a weakly simple polygon $ P_2 $, and that this ratio is larger than $ \frac{2}{\sqrt{3}} $ (Otherwise, $ \frac{\lVert X(s,t)\rVert}{\lVert \mathit{SP_w}(s,t)\rVert} \leq \frac{2}{\sqrt{3}} $, and we are done.).
	
	In the following, we replace a given instance by a $ P_2 $-triple, corresponding to a weakly simple polygon of type $P_2$ that has the same ratio as the given instance.
	Thus, instead of upper-bounding the ratio of the given instance, we can do it for the $P_2$-triple, which is substantially easier.

    \begin{lemma}
        \label{lem:triple}
        For any weakly simple polygon $ P' $ of type~$ P_2 $, a $ P_2$-triple can be defined such that $ P' $ is the polygon of type~$ P_2 $ of the vertex $ T_2 \cap T_3 \cap T_4 $, and the weights $ \omega_2, \omega_3, \omega_4 $ remain as in $ P' $. We say that this $ P_2$-triple \emph{corresponds} to the weakly simple polygon $ P'$.
    \end{lemma}

    \begin{proof}
        Let $ u_j $ and $ u_{j+1} $ be the points where $ \mathit{SP_w}(s,t) $ respectively enters and leaves~$ P' $, and let $ u_j, u_{j+1} \in T_{j'} $.
        Consider a $ P_2$-triple intersected by a path $ \mathit{SP_w}(s',t') $ where the weights of the cells $ T_2, \ T_3 $, and $ T_4 $ are, respectively, the same as the weights of cells $ T_{j'-1} $, $ T_{j'} $ and $ T_{j'+1} $ in the former instance. The weights of the cells $ T_1 $ and $ T_5 $ are obtained by solving the system of equations given by Snell's law of refraction. By construction, this $ P_2 $-triple is a valid instance of $ \mathit{SP_w}(s',t') $, and it intersects the cell $ T_3 $ forming the same weakly simple polygon~$ P' $ as in the former instance.
    \end{proof}
	
	Lemma~\ref{lem:triple} allows us to assume, from now on, that an upper bound on the ratio $ \frac{\lVert X(s,t)\rVert}{\lVert \mathit{SP_w}(s,t)\rVert} $ is given by the length of a path that intersects a $ P_2 $-triple as in Definition~\ref{def:triple}.
		
		\begin{lemma}
			\label{lem:8}
			Let $ \frac{\lVert\mathit{SGP_{w}}(s,t)\rVert}{\lVert \mathit{SP_{w}}(s,t)\rVert} $ be the maximum ratio attained by any path intersecting a polygon of type~$ P_2$. Consider its corresponding $ P_2$-triple. Let $ T_3 $ be the cell where the maximum ratio $ \frac{\lVert X(u_j, u_{j+1})\rVert}{\lVert \mathit{SP_w}(u_j, u_{j+1}) \rVert}, \ u_j, u_{j+1} \in T_3 $ is attained. Then, $ \lVert X(s, t) \rVert = \lVert \Pi_3(s,t) \rVert $.
		\end{lemma}
		
		\begin{proof}
		    We prove the result by contradiction, arguing that if there is at least one grid path $\mathit{GP_w}(s, t)$ among $ \{X(s, t), \Pi_3(s,t)\} $, that is strictly shorter than the other grid path, then this instance cannot maximize $ \frac{\lVert \mathit{SGP_w}(s,t)\rVert}{\lVert \mathit{SP_w}(s,t)\rVert} $.
		    We will show this by finding another assignment of weights $ w' $ for the cells $ T_1, \ldots, T_5 $, such that $ \frac{\lVert\mathit{GP_{w'}}(s,t)\rVert}{\lVert \mathit{SP_{w'}}(s,t)\rVert} > \frac{\lVert\mathit{GP_w}(s,t)\rVert}{\lVert \mathit{SP_w}(s,t)\rVert} $, proving that the given instance does not provide the maximum ratio.
		
		    We first set the weight of all the cells that are not traversed by~$ X(s,t) $ to infinity. This way, we ensure that when modifying the weights of some cells, the combinatorial structure of the shortest path is preserved. Let $ T_{\ell} $ be the cell that shares the edge of $ \Pi_3(s,t) $ with $ T_{3} $, see Figure~\ref{fig:29}. Recall that $ \omega_i $ is the weight of the triangular cell $ T_i $. Then, the weighted length of the crossing path $ X(s,t) $ along the edges of $ T_3 $ is $ 2\min\{\omega_{2}, \omega_3\} + 2\min\{\omega_3, \omega_{4}\} $, and the weighted length of the shortcut path $ \Pi_3(s,t) $ along the edges of $ T_3 $ is $ 2\min\{\omega_{3}, \omega_{\ell}\} = 2\omega_3 $ (because $\omega_{\ell}=\infty$). Let $ u_j $ and $ u_{j+1} $ be two consecutive points where $ X(s,t) $ and $ \mathit{SP_{w}}(s,t) $ coincide. Let $ [p, q] $, and $ [p', q'] $ be, respectively, the edges containing $ u_j $ and $ u_{j+1} $, where $ p, p' \in T_{\ell} $.
		
		    \begin{itemize}
		        \item If $\mathit{GP_w}(s,t) = X(s,t) $ then $ \lVert X(s,t) \rVert < \lVert \Pi_3(s,t) \rVert $, and we have that
		        \begin{equation}
		        \label{eq:minima_inequality}
		            \min\{\omega_{2}, \omega_3\} + \min\{\omega_{3}, \omega_{4}\} < \omega_3.
		        \end{equation}
		        \begin{itemize}
		            \item If $ \omega_{3} \leq \omega_{2} $, then $ \omega_3 + \min\{\omega_{3}, \omega_{4}\} < \omega_3 $, which is not possible since $ \min\{\omega_{3}, \omega_{4}\} \geq 0 $. Hence, $ \omega_{3} > \omega_{2} $.
		            \item If $ \omega_{3} \leq \omega_{4} $, then $ \min\{\omega_{2}, \omega_{3}\} + \omega_{4} < \omega_3 $, which is not possible since $ \min\{\omega_{2}, \omega_{3}\} \geq 0 $. Hence, $ \omega_{3} > \omega_{4} $.
		        \end{itemize}
		        These two facts together with Equation~(\ref{eq:minima_inequality}) imply that $ \omega_{2} + \omega_{4} < \omega_3 $. We also have that 
		        \begin{equation*}
		           \frac{\lVert \mathit{GP_w}(s,t) \rVert}{\lVert \mathit{SP_w}(s,t) \rVert} = \frac{\lVert X(s,p) \rVert + 2(\omega_{2} + \omega_{4}) + \lVert X(p',t) \rVert}{\lVert \mathit{SP_w}(s,u_{j}) \rVert + |u_ju_{j+1}|\omega_3 + \lVert \mathit{SP_w}(u_{j+1},t) \rVert}
		        \end{equation*}
		        and we know that $ |u_ju_{j+1}| > 0 $, so we can decrease the weight $ \omega_3 $, increasing the ratio.
		
		        \item Otherwise, if $\mathit{GP_w}(s,t) = \Pi_3(s,t) $ then $ \lVert \Pi_3(s,t) \rVert < \lVert X(s,t) \rVert $, and we have that $ \omega_3 < \min\{\omega_{2}, \omega_3\} + \min\{\omega_{3}, \omega_{4}\} $. We also have that
		        \begin{equation*}
		           \frac{\lVert \mathit{GP_w}(s,t) \rVert}{\lVert \mathit{SP_w}(s,t) \rVert} = \frac{\lVert \Pi_3(s,p) \rVert + 2\omega_{3} + \lVert \Pi_3(p',t) \rVert}{\lVert \mathit{SP_w}(s,u_{j}) \rVert + \lVert \mathit{SP_w}(u_{j}, u_{j+1}) \rVert + \lVert \mathit{SP_w}(u_{j+1},t) \rVert}.
		        \end{equation*}
		
		        The ratio $ \frac{\lVert \mathit{GP_w}(s,t) \rVert}{\lVert \mathit{SP_w}(s,t) \rVert} $ is a strictly monotonic function for every $ \omega_k $ \cite{chew1984pseudolinearity, RAPCSAK1991353}. So, if this function is decreasing in the direction of $ \omega_3 $, we can decrease the weight $ \omega_3 $. Otherwise, we can increase the weight $ \omega_3 $. In both cases, we can increase the ratio.
		
		    \end{itemize}
		    Thus, we found another weight assignment $ w' $ such that $ \frac{\lVert \mathit{GP_{w'}}(s,t) \rVert}{\lVert SP_{w'}(s,t) \rVert} > \frac{\lVert \mathit{GP_w}(s,t) \rVert}{\lVert \mathit{SP_{w}}(s,t) \rVert} $. In addition, the change in~$\omega_3$ can be as small as needed so that the weighted length of $ \mathit{GP_{w'}}(s,t)$ is not larger than that of the other grid path in the set $ \{X(s, t), \Pi_3(s,t)\} $ with the new weight assignment~$w'$ and, thus, it does not change which of the two grid paths in the set is the shorter one.
		\end{proof}
		
		Now, we have all the tools needed to obtain an upper bound on the ratio $ \frac{\lVert X(u_j, u_{j+1})\rVert}{\lVert \mathit{SP_w}(u_j, u_{j+1}) \rVert} $ for a $P_2$-triple.
		Lemma~\ref{lem:26} presents an upper bound on this ratio where $ u_j, u_{j+1} \in T_3 $ are two consecutive points where $ X(s,t) $ and $ \mathit{SP_w}(s, t) $ coincide, and $ \lVert X(s,t) \rVert = \lVert \Pi_3(s,t) \rVert $. Since the exact shape of $ \mathit{SP_w}(s,t) $ is unknown, when computing the ratio in Lemma~\ref{lem:26}, we will maximize the ratio for any position of the points $ u_j $ and $ u_{j+1} $.
		
		\begin{lemma}
			\label{lem:26}
			Let $ u_j, u_{j+1} \in T_3 $, be two consecutive points where a shortest path~$ \mathit{SP_w}(s,t) $ and the crossing path $ X(s,t) $ coincide. If $ u_j, \ u_{j+1} $ induce a weakly simple polygon of type~$ P_2 $, and $ \lVert X(s,t) \rVert = \lVert \Pi_3(s,t) \rVert $, then $ \frac{\lVert X(u_j, u_{j+1})\rVert}{\lVert \mathit{SP_w}(u_j, u_{j+1}) \rVert} \leq \frac{2}{\sqrt{3}} $.
		\end{lemma}
		
		\begin{proof}
			 Let $ (v^3_1, v^3_2, v^3_3) $ be the sequence of vertices on the boundary of $ T_3 $ in clockwise order. Since $ u_j $ and $ u_{j+1} $ are two consecutive points where $ \mathit{SP_w}(s,t) $ and $ X(s,t) $ coincide, and they induce a polygon of type~$ P_2 $, $ \mathit{SP_w}(s,t) $ enters $ T_3 $ from cell $ T_{2} $ through $ u_j $, and $ \mathit{SP_w}(s,t) $ leaves~$ T_3 $ and enters cell~$ T_{4} $ through~$ u_{j+1} $, see Figure~\ref{fig:65}. Suppose, without loss of generality, that $ u_j \in [v^3_1, v^3_2] $ and $ u_{j+1} \in [v^3_2, v^3_3] $. Let $ a, b, c $ be the lengths $ |u_jv^3_2|, |v^3_2u_{j+1}| $, and $ |u_ju_{j+1}| $, respectively. According to Observation~\ref{obs:3}, $ c = \sqrt{a^2+b^2-ab} $. We want to maximize the ratio $ \frac{\lVert X(u_j, u_{j+1})\rVert}{\lVert \mathit{SP_w}(u_j, u_{j+1}) \rVert} $ for all weight assignments of $ \omega_{2}, \omega_3 $, and $ \omega_{4} $.
			
			 \begin{figure}[tb]
			     \centering
	    	    \includegraphics{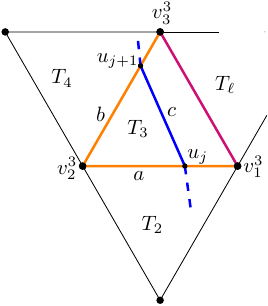}
    	        \caption{Subpaths of the grid path $ \Pi_3(s,t) $ (purple), the crossing path $ X(s,t) $ (orange), and $ \mathit{SP_w}(s,t) $ (blue) traversing a polygon of type~$ P_2 $ in a triangular tessellation.}
    	        \label{fig:65}
	\end{figure}
			
			 Traversing cell $ T_3 $ there is also the grid path $ \Pi_3(s,t) $. Let $ T_{\ell} $ be the cell that shares the edge of $ \Pi_3(s,t) $ with $ T_3 $. Since $ \lVert X(s,t)\rVert = \lVert \Pi_3(s,t)\rVert$, $ 2\min\{\omega_{2}, \omega_3\}+2\min\{\omega_3, \omega_{4}\} = 2\min\{\omega_3, \omega_{\ell}\} $. Thus, we distinguish two cases: $ b \leq a $ and $ b > a $. In the first case we take:
			\begin{equation*}
				\min\{\omega_{2}, \omega_3\} = \min\{\omega_3, \omega_{\ell}\} - \min\{\omega_3, \omega_{4}\}.
			\end{equation*}
			In the second case we take:
			
			\begin{equation*}
				\min\{\omega_3, \omega_{4}\} = \min\{\omega_3, \omega_{\ell}\} - \min\{\omega_{2}, \omega_3\}.
			\end{equation*}
			
			Then, an upper bound on the ratio $ R = \frac{\lVert X(u_j, u_{j+1})\rVert}{\lVert \mathit{SP_w}(u_j, u_{j+1})\rVert} $ in a polygon of type~$ P_2 $ is:
			
			\begin{align*}
				R & =
				\frac{a\min\{\omega_{2}, \omega_{3}\}+b\min\{\omega_3, \omega_{4}\}}{c\omega_3} \leq
				\begin{cases}
					\overbrace{\leq}^{\text{if } b \leq a} \frac{a\min\{\omega_3, \omega_{\ell}\}+(b-a)\min\{\omega_3, \omega_{4}\}}{c\omega_3} \leq \\
					\overbrace{\leq}^{\text{if } a < b} \frac{b\min\{\omega_3, \omega_{\ell}\}+(a-b)\min\{\omega_{2}, \omega_3\}}{c\omega_3} \leq
				\end{cases} \\
				& \leq \begin{cases}
					\frac{a\min\{\omega_3, \omega_{\ell}\}}{c\omega_3} \leq \frac{a}{\sqrt{a^2+b^2-ab}} \\
					\frac{b\min\{\omega_3, \omega_{\ell}\}}{c\omega_3} \leq \frac{b}{\sqrt{a^2+b^2-ab}}
				\end{cases} \leq \frac{2}{\sqrt{3}},
    		\end{align*}
    		where the last inequality in the two ratios is obtained by maximization over the values of~$ a \in [0,2] $ and $ b \in [0,2] $.
		\end{proof}

    Finally, we have all the pieces to prove our main result.

	\begin{theorem}
		\label{thm:6}
		In $ G_{6\text{corner}}, \ \frac{\lVert \mathit{SGP_w}(s, t)\rVert}{\lVert \mathit{SP_w}(s,t) \rVert} \leq \frac{2}{\sqrt{3}} $.
	\end{theorem}
	
	\begin{proof}
		Let $ \mathit{SP_w}(s,t) $ be a weighted shortest path between two corners $ s $ and $ t $ in a triangular tessellation. Let $ X(s,t) $ be the crossing path from $ s $ to $ t $ obtained from $ \mathit{SP_w}(s,t) $. By Observation~\ref{thm:1}, $ \frac{\lVert X(s,t)\rVert}{\lVert \mathit{SP_w}(s,t)\rVert} \leq \frac{\lVert X(u_j,u_{j+1})\rVert}{\lVert \mathit{SP_w}(u_j,u_{j+1})\rVert} $, over all pairs $ (u_j,u_{j+1}) $ of consecutive points where $ \mathit{SP_w}(s,t) $ and $ X(s,t) $ coincide.
		
		As already observed, the ratio $ \frac{\lVert X(u_j, u_{j+1})\rVert}{\lVert \mathit{SP_w}(u_j, u_{j+1}) \rVert} $ is $ 1 $ in a polygon of type~$ P_1 $. Further, by Lemma~\ref{lem:27}, that ratio is at most $ \frac{2}{\sqrt{3}} $ for weakly simple polygons of type~$ P_k, \ k > 2 $. Finally, using Lemmas \ref{lem:triple} and \ref{lem:8}, we know that if a path intersecting cell $ T_3 $ is the path that maximizes the ratio $ \frac{\lVert X(u_j, u_{j+1})\rVert}{\lVert \mathit{SP_w}(u_j, u_{j+1}) \rVert} $ in a weakly simple polygon of type~$ P_2$, then $ \lVert X(s,t)\rVert = \lVert \Pi_3(s,t)\rVert $. And, in this case, by Lemma~\ref{lem:26} we get that the ratio $ \frac{\lVert X(u_j,u_{j+1})\rVert}{\lVert \mathit{SP_w}(u_j,u_{j+1}) \rVert} $ is at most $ \frac{2}{\sqrt{3}} $, where $ u_j, u_{j+1} \in T_3 $.
		
		All this implies that $ \frac{\lVert X(s,t)\rVert}{\lVert \mathit{SP_w}(s,t) \rVert} $ is at most $ \frac{2}{\sqrt{3}} $. Since $ \lVert \mathit{SGP_w}(s,t)\rVert \leq \lVert X(s,t)\rVert $, we have that $ \frac{\lVert \mathit{SGP_w}(s, t)\rVert}{\lVert \mathit{SP_w}(s,t) \rVert} \leq \frac{2}{\sqrt{3}} $.
	\end{proof}
	
	Figure~\ref{fig:67} provides an illustration of the lower bound $ \frac{2}{\sqrt{3}} $ on the ratio between the weighted shortest grid path $ \mathit{SGP_w}(s,t) $ (red) and the weighted shortest path $ \mathit{SP_w}(s,t) $ (blue) claimed by Nash~\cite{Nash}. Hence, the upper bound in Theorem~\ref{thm:6} is tight for $ G_{6\text{corner}} $.

    \begin{figure}[tb]
	    \centering
	    \includegraphics{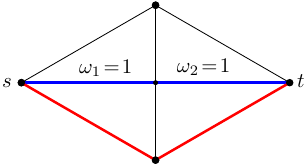}
	    \caption{$ \mathit{SP_w}(s,t) $ and $ \mathit{SGP_w}(s,t) $ are depicted in blue and red, respectively. The ratio $ \frac{\lVert \mathit{SGP_w}(s,t)\rVert}{\lVert \mathit{SP_w}(s,t)\rVert} $ is $ \frac{2}{\sqrt{3}} $.}
	    \label{fig:67}
    \end{figure}

\section{Ratios $ \frac{\lVert \mathit{SGP_w}(s,t)\rVert}{\lVert \mathit{SVP_w}(s,t)\rVert} $ in $ G_{6\text{corner}} $ and $ \frac{\lVert \mathit{SVP_w}(s,t)\rVert}{\lVert \mathit{SP_w}(s,t)\rVert} $ in $ G_{\text{corner}} $}
	
	In this section we provide results for the ratios where the weighted shortest vertex path $ \mathit{SVP_w}(s,t) $ is involved, i.e., $ \frac{\lVert \mathit{SGP_w}(s,t)\rVert}{\lVert \mathit{SVP_w}(s,t)\rVert} $ and $ \frac{\lVert \mathit{SVP_w}(s,t)\rVert}{\lVert \mathit{SP_w}(s,t)\rVert} $. The length of a weighted shortest vertex path $ \mathit{SVP_w}(s,t) $ is an upper bound for the length of a weighted shortest path $ \mathit{SP_w}(s,t) $, so the upper bound on the ratio $ \frac{\lVert \mathit{SGP_w}(s, t)\rVert}{\lVert \mathit{SP_w}(s,t) \rVert} $ obtained in Theorem~\ref{thm:6} is an upper bound for $ \frac{\lVert \mathit{SGP_w}(s, t)\rVert}{\lVert \mathit{SVP_w}(s,t) \rVert} $.
	
	\begin{corollary}
		\label{cor:6}
		In $ G_{6\text{corner}}, \ \frac{\lVert \mathit{SGP_w}(s, t)\rVert}{\lVert \mathit{SVP_w}(s,t) \rVert} \leq \frac{2}{\sqrt{3}} $.
	\end{corollary}
	
	When the weights of the cells are in the set $ \{1, \infty\} $, the ratio $ \frac{\lVert \mathit{SGP_w}(s, t)\rVert}{\lVert \mathit{SVP_w}(s,t) \rVert} $ was proved to be at most $ \frac{2}{\sqrt{3}} $ by Nash~\cite{Nash}. In addition, Nash showed that this bound is tight. Thus, for general (non-negative) weights this value is a lower bound on the ratio $ \frac{\lVert \mathit{SGP_w}(s, t)\rVert}{\lVert \mathit{SVP_w}(s,t) \rVert} $ for $ G_{6\text{corner}} $.
	
	
	As a corollary of Theorem~\ref{thm:6}, we obtain Corollary~\ref{cor:5}. The result comes from the fact that $ \lVert \mathit{SVP_w}(s,t) \rVert $ is a lower bound for $ \lVert \mathit{SGP_w}(s,t) \rVert $. Recall that $ \mathit{SVP_w}(s, t) $ and $ \mathit{SP_w}(s, t) $ do not use $ G_{6\text{corner}} $, but $ G_{\text{corner}} $.
	
	\begin{corollary}
		\label{cor:5}
		In $ G_{\text{corner}} $, $ \frac{\lVert \mathit{SVP_w}(s, t)\rVert}{\lVert \mathit{SP_w}(s,t) \rVert} \leq \frac{2}{\sqrt{3}} \approx 1.15 $.
	\end{corollary}
	
	\begin{figure}[tb]
	    \centering
	    \includegraphics{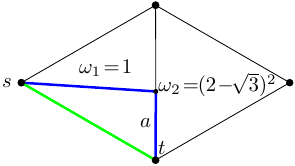}
	    \caption{$ \mathit{SP_w}(s,t) $ and $ \mathit{SVP_w}(s,t) $ are depicted in blue and green, respectively. The ratio $ \frac{\lVert \mathit{SVP_w}(s,t)\rVert}{\lVert \mathit{SP_w}(s,t)\rVert} $ is $ \frac{2\sqrt{7\sqrt{3}-12}}{(7-4\sqrt{3})(6\sqrt{2}+\sqrt{7\sqrt{3}-12})} $ when $ a = \frac{4-(4\sqrt{6}+7\sqrt{2})\sqrt{1351\sqrt{3}-2340}}{4} \approx 0.87 $.}
	    \label{fig:31}
    \end{figure}
	
	Finally, we provide a lower bound for the ratio $ \frac{\lVert \mathit{SVP_w}(s, t)\rVert}{\lVert \mathit{SP_w}(s,t) \rVert} $. The green path in Figure~\ref{fig:31} is a weighted shortest vertex path~$ \mathit{SVP_w}(s,t) $ between vertices $ s $ and $ t $, thus, we have the following result.
	
	\begin{observation}
	    \label{obs:5}
	    In $ G_{\text{corner}} $, $ \frac{\lVert \mathit{SVP_w}(s, t)\rVert}{\lVert \mathit{SP_w}(s,t) \rVert} \geq  \frac{2\sqrt{7\sqrt{3}-12}}{(7-4\sqrt{3})(6\sqrt{2}+\sqrt{7\sqrt{3}-12})} \approx 1.11 $.
	\end{observation}

\section{Discussion and future work}

We presented bounds on the ratio between the lengths of three types of weighted shortest paths in a triangular tessellation.
The fact that a compact grid graph such as $G_{6\text{corner}}$ guarantees an error bound of $\approx 15\%$, regardless of weights used, justifies its widespread use in applications in areas such as gaming and simulation, where performance is a priority over accuracy.

Our analysis techniques, presented here for triangular grids, can also be applied to obtain upper bounds for the same ratios in the other two types of regular tessellations, square and hexagonal. In particular, we recently proved upper bounds of $ R = \frac{2}{\sqrt{2+\sqrt{2}}} $ for weighted square cells \cite{SquaresEuroCG}, and $ R = \frac{3}{2} $ when we only allow movement along the edges of an hexagonal tessellation~\cite{HexagonsYRF}.
The main differences lie in the exact definition of the crossing paths and the weakly simple polygons.
Our techniques can also be used to derive upper bounds for another type of grid graph, where the vertices are cell centers instead of corners (see, e.g.,~\cite{bailey2015path,Bound3,Nash}).

For future work, it would be interesting to close the gap for $\frac{\lVert \mathit{SVP_w}(s,t)\rVert}{\lVert \mathit{SP_w}(s,t)\rVert}$, of approximately $ 0.04 $.
It is an intriguing question whether the seemingly richer any-angle path $\mathit{SVP_w}(s,t)$ can actually guarantee a better quality factor than $G_{6\text{corner}}$. However, our results show that, even if that is the case, the improvement is rather negligible.

\vspace{15mm}

{\small \noindent \textbf{Acknowledgments}}

P. B. is partially supported by NSERC. G. E., D. O. and R. I. S. are partially supported by H2020-MSCA-RISE project 734922 - CONNECT and project PID2019-104129GB-I00 funded by MCIN/AEI/10.13039/501100011033. G. E. and D. O. are also supported by PIUAH21/IA-062 and CM/JIN/2021-004. G. E. is also funded by an FPU of the Universidad de Alcal\'a.

The authors want to thank Vsevolod Kvachev for granting permission to use the screenshots of the game ``Colossal Citadels''.

The authors thank the anonymous referees for their comments and insights which greatly improved the readability of this article.




\end{document}